\documentclass{ws-m3as}

\begin{document}

\markboth{Xuejiao Liu, Yu Qiao and Benzhuo Lu}{}
%\markboth{Xuejiao Liu, Yu Qiao, Benzhuo Lu}{Analysis of the Mean Field Free Energy Functional of Electrolyte Solution with Non-zero Boundary Conditions \\and the Generalized PB/PNP Equations with Inhomogeneous Dielectric Permittivity}

%%%%%%%%%%%%%%%%%%% Publisher's Area please ignore %%%%%%%%%%%%%%%%%%%%%%%
%
%\catchline{}{}{}{}{}
%
%%%%%%%%%%%%%%%%%%%%%%%%%%%%%%%%%%%%%%%%%%%%%%%%%%%%%%%%%%%%%%%%%%%%%%%%%%

\title{Analysis of the Mean Field Free Energy Functional of Electrolyte Solution with Non-zero Boundary Conditions and the Generalized PB/PNP Equations with Inhomogeneous Dielectric Permittivity}

%\author{Xuejiao Liu\footnotemark[1] , Yu Qiao\footnotemark[2] and Benzhuo Lu\footnotemark[3]\footnotetext[3]{Corresponding author}}
\author{Xuejiao Liu\footnotemark[1], Yu Qiao\footnotemark[1], and Benzhuo Lu\footnotemark[1]\footnotemark[2]}
\footnotetext[1]{State Key Laboratory of Scientific and Engineering Computing, National Center for Mathematics and Interdisciplinary Sciences, Academy of Mathematics and Systems Science, Chinese Academy of Sciences, Beijing 100190, China (liuxuejiao@lsec.cc.ac.cn, qiaoyu@lsec.cc.ac.cn, bzlu@lsec.cc.ac.cn).}
\footnotetext[2]{Corresponding author.}
%\address{State Key Laboratory of Scientific and Engineering Computing, National Center for Mathematics and Interdisciplinary Sciences, Academy of Mathematics and Systems Science, Chinese Academy of Sciences, Beijing 100190, China\\
%\footnotemark[1] liuxuejiao@lsec.cc.ac.cn\\
%\footnotemark[2] qiaoyu@lsec.cc.ac.cn\\
%\footnotemark[3] bzlu@lsec.cc.ac.cn}

\maketitle

\begin{history}
%\received{(Day Month Year)}
%\revised{(Day Month Year)}
%\accepted{(Day Month Year)}
%\comby{(xxxxxxxxxx)}
\end{history}

%\begin{abstract}
\textbf{Abstract.} The energy functional, the governing partial differential equation(s) (PDE), and the boundary conditions need to be consistent with each other in a modeling system.
In electrolyte solution study, people usually use a free energy form
of an infinite domain system (with vanishing potential boundary condition) and the derived PDE(s) for analysis and computing.
However, in many real systems and/or numerical computing, the objective
domain is finite, and people still use the similar energy form, PDE(s) but with
different boundary conditions, which may cause inconsistency.
In this work, (1) we present a mean field free energy functional for electrolyte solution within a finite domain with either physical or numerically required artificial boundary. Apart from the conventional energy components (electrostatic potential energy, ideal gas entropy term and chemical potential term), new boundary interaction terms are added for both Neumann and Dirichlet boundary conditions. These new terms count for physical interactions with the boundary (for real boundary) or the environment influence on the computational domain system (for non-physical but numerically designed boundary).
(2) The traditional physical-based Poisson-Boltzmann (PB) equation and Poisson-Nernst-Planck (PNP) equations are proved to be consistent with the new free energy form, and different boundary conditions can be applied.
(3) In particular, for inhomogeneous electrolyte with ionic concentration-dependent dielectric permittivity, we derive the generalized Boltzmann distribution (thereby the generalized PB equation) for equilibrium case, and the generalized PNP equations for non-equilibrium case, under different boundary conditions.
Numerical tests are performed to demonstrate the different consequences resulted from different energy forms and their derived PDE(s).
%\end{abstract}



%\ccode{AMS Subject Classification: 22E46, 53C35, 57S20}

\section{Introduction}
As a requirement both in physics and mathematics, the system energy functional,
the governing partial differential equation(s) (PDE), and the boundary condition(s) (BD) need to be consistent.
People usually derive the PDE(s) through minimization of a free energy functional $F$, in which the information of
boundary condition(s) associated with the PDE is in principle included.
However, a common case is that once a type of PDE is obtained (usually from an energy functional for a infinite system),
people may study, either on theoretically or numerically,
the PDE under different boundary conditions. But in this case the changed boundary condition may be inconsistent with
the original energy form, and may cause unreasonable results.
An example is the electrolyte system, which is the focus of this work.

Electrolyte solution is a charged system mixed with polarizable solvent and mobile ions, which exists
in many areas such as chemistry, colloid, fuel cell, material science, and biology systems.
Enormous amount of literatures can be found in this area. In mean field theory, a Poisson-Boltzmann (PB) equation
is a physically reasonable description of the equilibrium state of electrolyte solution. In non-equilibrium state
(i.e., non-balanced ionic flow exists), the Poisson-Nernst-Planck (PNP) equations is a proper model to describe the coupling
of ionic diffusion processes and the generated electric field.
The PB equation and PNP equations are two most commonly used PDEs in electrolyte solution system.
These equations can also be derived from variation of the free energy.
Sharp and Honig have used the calculus of variations to provide a unique definition of the total energy and
to obtain expressions for the total mean field electrostatic free energy of electrolyte solution (including fixed macromolecules)
for both linear and nonlinear PB equations,\cite{Sharp90} and later Gilson et al. derived the mean forces based on
mean field electrostatic free energies.\cite{Gilson93}
\begin{equation}\label{1.1}
    F = \int\{\rho^{f}\phi - \frac{1}{2}\epsilon|\nabla\phi|^{2} - \beta^{-1}\sum_{i=1}^{K}c_{i}^{\infty}(e^{-\beta q_{i}\phi}-1)\}dV.
\end{equation}
And in turn, the PBE can also be expected to be derived from these energy functionals.
Gilson et al. have shown that if the free energy $F$ is considered as a functional with respect to (w.r.t.) the potential function,
the potential which extremizes $F$ is also the potential that satisfies the Poisson-Boltzmann equation.\cite{Gilson93}
Fogolari and Briggs have pointed out that the potential satisfying the PBE in fact maximizes the energy functional if it
is considered as a functional w.r.t potential.\cite{Fogolari97} When the free energy functional is regarded as functional w.r.t the concentration $c$ rather than the potential $\phi$, they proved that the PB distribution is then the only distribution which minimizes the free energy (the Poisson is considered as a constraint).\cite{Fogolari97}
This conclusion was also re-stated in a more mathematical way later.\cite{LiBo09b}
The energy functional takes form
\begin{equation}\label{1.2}
    F = \int_{\Omega}\frac{1}{2}\rho\phi dV
       + \beta^{-1}\sum_{i=1}^{K}\int_{\Omega}c_{i}[\log(\Lambda^{3}c_{i})-1]dV - \sum_{i=1}^{K}\int_{\Omega}\mu_{i}c_{i}dV,
\end{equation}
with a Poisson equation as a constraint.
Another advantage of this form is that this form can be applied to study of both equilibrium and non-equilibrium state
of the electrolyte solution.
It is worth noting that those free energy forms are for electrolyte solution in an infinite domain
where the potential (and the derivative) goes to zero at the boundary.
However, a real physical system and/or a practically computational domain (as appeared in finite element/finite difference methods)
are often finite, and the boundary conditions are usually non-trivial and non-zero.
In electrokinetics, most physically interesting properties arise from different non-zero boundary conditions .\cite{BTu13,Shixin13,Chapman13,Lee11,Xu11}
In these non-zero BD cases for charged system, the system's free energy also needs to include
the physical interaction between the system and the boundary. As a consequence in mathematical analysis,
these additional boundary energy terms also need to appear in the energy functional.
In other words, the traditional PB equations with general non-zero Dirichlet or/and Neumann BDs
can not be derived from above free energy form (either Eq.~\eqref{1.1} or \eqref{1.2})
because the boundary term(s) are missed in the energy functional. The issue will be solved in this work.
It is worth noting here that even if a real system is infinite, but in practical computation as in finite element
approach, only a finite domain is taken and certain non-trivial BD(s)
need to adopt to simulate the behaviour of the whole system. In this case, if we need a, for instance, non-zero Dirichlet BD,
an energy term needs to be included in the free energy and represent interaction between the system and the Dirichlet type of boundary.
This is physically reasonable, because the boundary interaction term can be an exact representation or proper approximation of the
interaction between the finite modeling system and the infinite outside part which is not involved in the computational
domain (see detailed physical explanations in the Theory section).
Therefore, in the rest of this article, we will not discriminate a boundary as a physical (interfacial) boundary or as an artificial
boundary, as they will be treated similarly in the energy form.

The free energy functional for an infinite electrolyte solution system can be considered as a special case under
zero-boundary condition at infinite boundary. If this energy functional is used to derive the PDE with non-zero BD,
it may resulted "screwed" equation. Such an example can be found in a recent work.\cite{LiBo16}
A following non-zero Dirichlet boundary-value problem of Poisson's equation \eqref{1.3} is considered,
which is constraint of the potential $\phi$ in the traditional free energy functional,
\begin{align}\label{1.3}
  -\nabla\cdot(\epsilon\nabla\phi(c)) &= \rho(c)     &in \;\Omega,\\
  \epsilon\frac{\partial\phi}{\partial n} & = \sigma &on \;\Gamma_{N},\notag\\
  \phi & = \phi_{0} &on \;\Gamma_{D},\notag
\end{align}
where $\frac{\partial\phi}{\partial n}$ denotes the normal derivative at the boundary with $n$ the exterior unit normal.
In analysis, it generally needs to introduce a corresponding homogeneous boundary-value problem of Poisson's equation \eqref{1.4}
which has the unique weak solution $\phi_{D}$.
\begin{align}\label{1.4}
  \nabla\cdot(\epsilon\nabla\phi_{D}(c)) &= 0     & in \quad \Omega,\\
  \epsilon\frac{\partial\phi_{D}}{\partial n} & = 0& on \quad \Gamma_{N},\notag\\
  \phi_{D} & = \phi_{0} & on \quad \Gamma_{D}.\notag
\end{align}
Using variational approach to the free energy functional with incomplete boundary terms can lead to
a "screwed" Boltzmann distribution and an unusual PB equation.
Similarly, for non-equilibrium state and inhomogeneous boundary-value problem, we will show details in following sections that
applying variational approach to the incomplete free energy functional will lead to a set of different PNP equations from
the traditionally established one (supposing $\epsilon$ is constant):
\begin{equation}\label{1.5}
    -\nabla\cdot\epsilon\nabla\phi(c) = \rho^{f} + \lambda\sum_{i=1}^{K}q_{i}c_{i}, \quad in \;\Omega
\end{equation}
\begin{equation}\label{1.6}
    \frac{\partial c_{i}}{\partial t}=\nabla \cdot (D_{i}[\nabla c_{i}+\beta c_{i}\nabla(q_{i}[\phi(c)-\frac{1}{2}\phi_{D}(c)])]),
    \quad in \;\Omega_{s}, i=1, 2, \cdots, K.
\end{equation}
In the physics of electro-diffusion process and in the traditional PNP equations, the drift term $\beta q c\nabla\phi$ is determined by
the electric field, i.e. $\nabla\phi$ and should be irrelevant to $\phi_{D}$. But in Eqs.~\eqref{1.5} and \eqref{1.6},
an additional term $-{1 \over 2}\beta q c\nabla \phi_D$ appears in the drift term and is unavoidable in variational approach using the incomplete energy functional (see the Section 2).

To derive the correct PB and PNP equations subject to different BDs (Neumann, Dirichlet or their co-existing case),
we will provide in this paper a complete energy functional form, which is consistent with the PDEs and the corresponding BDs.
Furthermore, the energy functional is also shown to satisfy the energy dissipation law.
Numerical examples demonstrate significant deviations of the predictions from incorrect PB/PNP models (originated from incomplete
energy functional) from the right ones.

In addition, a particular interesting case of this work is to consider the situation that
dielectric coefficient is dependent on ionic concentration. The general free energy functional includes this situation
and variational approach is applied to derive the generalized PB and PNP equations under different boundary conditions.
Ionic solutions may be considered to consist of 3 constituents: the charged anions and cations, "hydration" solvent molecules
near the vicinity of the ions, and "free" solvent molecules. The hydration shells will affect the dielectric coefficient in an
ionic solution.\cite{Wei89,Wei90,Wei92} A lot of experiments and theoretical analysis have indicated that the dielectric coefficient decreased with the increase of local ionic concentrations.\cite{Hasted48,HO77,HO78,Nortemann97,Berk06,Lubos09,Hlli14} In our previous paper,\cite{Hlli14} we present a variable dielectric PB model for biological study, in which the dielectric coefficient is ionic concentration-dependent. However, the equation is not mathematically consistent with the system's free energy functional. In this paper, we analyse and discuss a general dependence form of the dielectric coefficient on local concentrations, and the governing equations in both equilibrium and non-equilibrium
are consistently given.

\section{Theory and Method}
\subsection{The mean field free energy functional}
We consider the general case of an electrolyte solution that contains solvent, arbitrary number of mobile ion species, and perhaps membrane-molecule(s) or nanopore as well. The macro-object like molecule, if exists,
is treated as fixed object and usually also carries charges inside or on the surface.
Figure 1 represents two typical biophysical models in computational and analysis.
The domain $\Omega_{s}$ denotes the solvent region where there is a mixed solution with diffusive ion species, such as mobile
ions. The solute region $\Omega_{m}$ is the domain occupied by (in (a)) the fixed biomolecule, such as protein or DNA,
or by (in (b)) the membrane, channel protein/nanopore.\cite{Bzlu07,Siwy06,BTu13,WGWei12} In case (b), if necessary,
$\Omega_m$ can be further divided into different sub-regions, but this does not affect our following analysis.
The whole domain is denoted by $\Omega = \Omega_{m} + \Omega_{s}$.

\begin{figure}[pb]
\centering
%\subfigure[The biomolecule system]{
%\label{fig.1}
\includegraphics[width=0.4\textwidth]{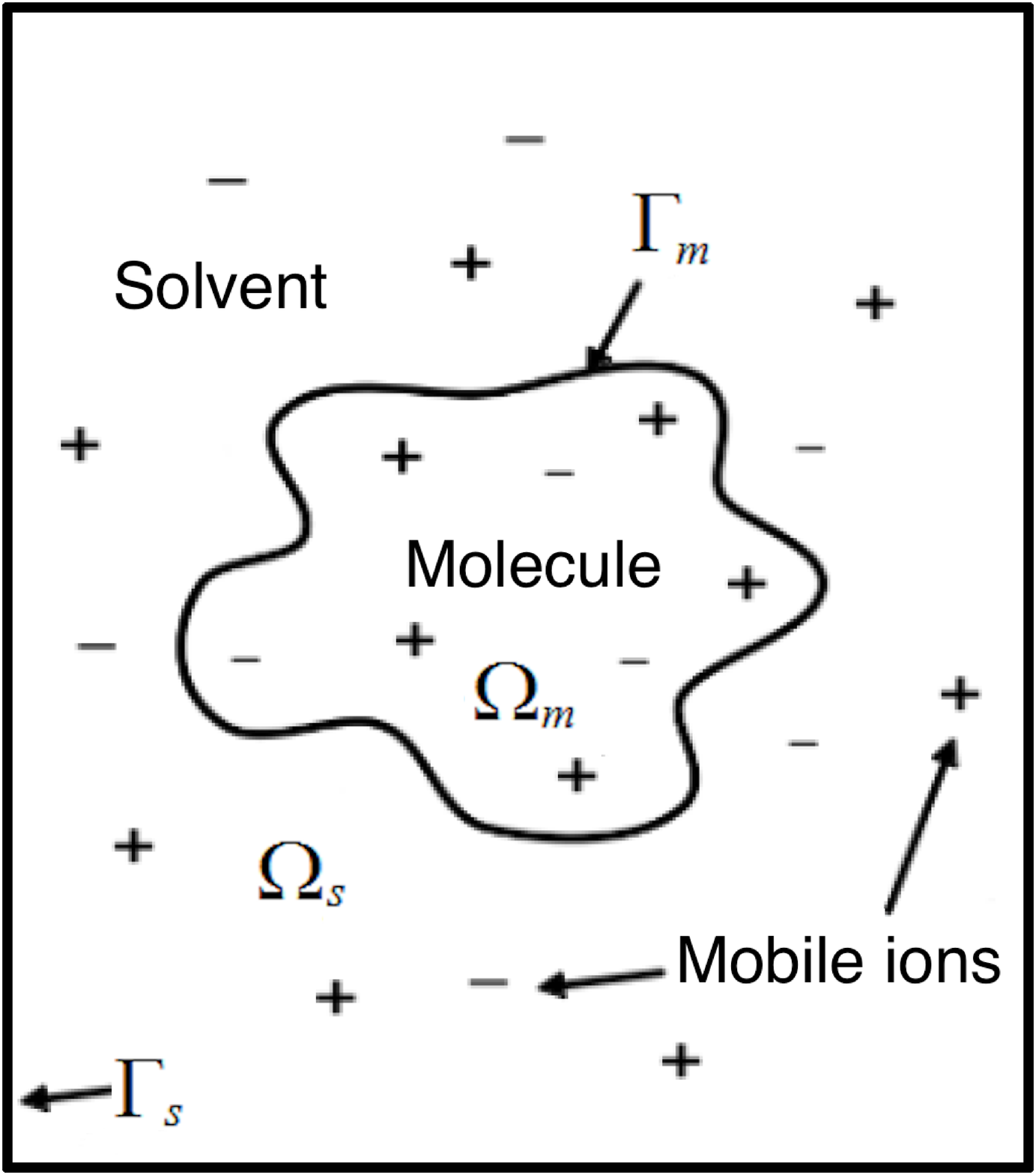}
%\subfigure[The ion channel system]{
%\label{fig.2}
\includegraphics[width=0.4\textwidth]{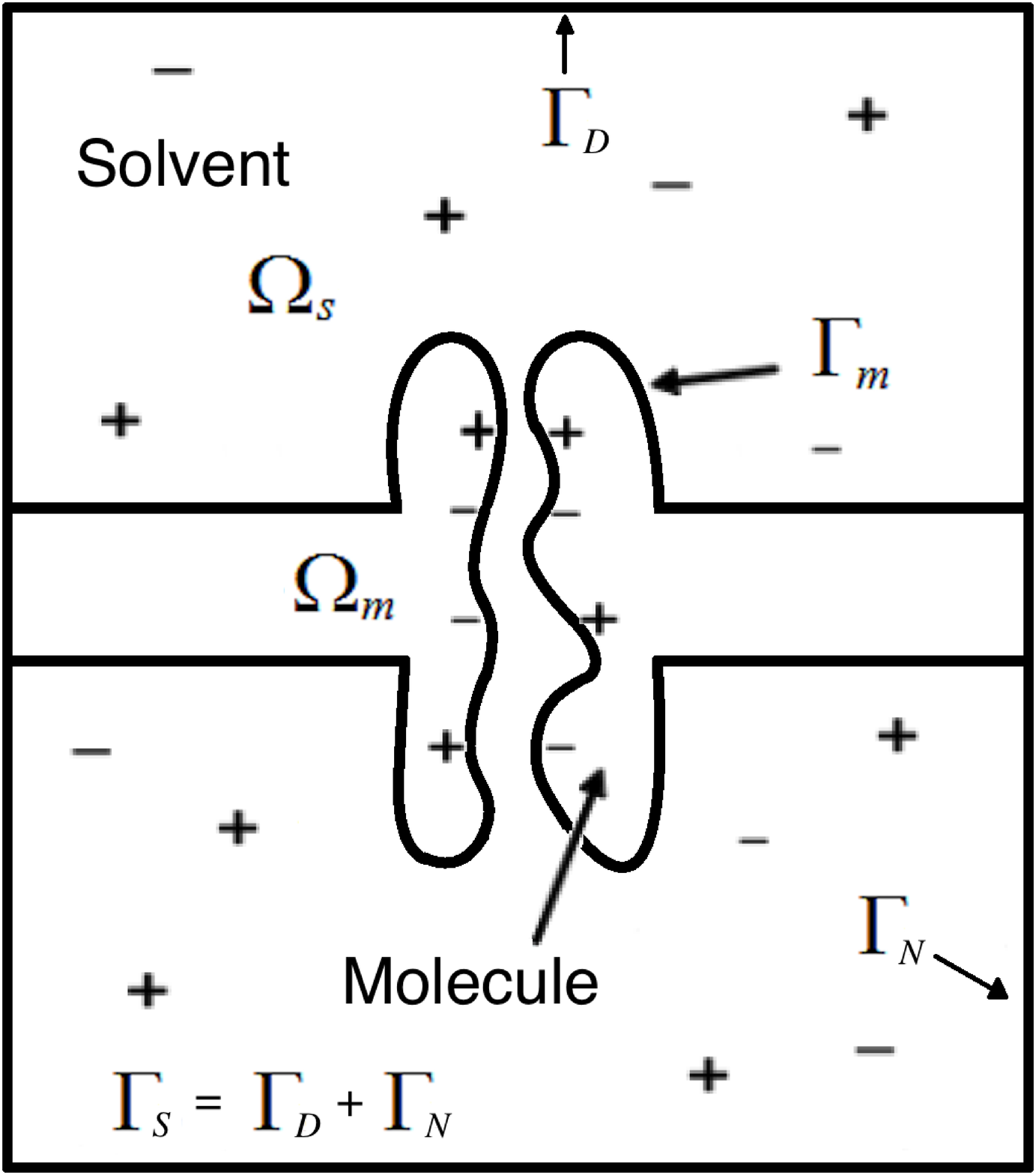}
%\centerline{\psfig{file=fig1_1.eps,width=1.5in}\psfig{file=fig2_1.eps,width=1.5in}}
%\centerline{\psfig{file=fig2_1.eps,width=1.5in}}
%\vspace*{8pt}
\caption{A 2-D schematic view of the ionic solution system: (a) with one fixed biomolecule;
(b) with an ion channel (or similar a nanopore) embedded in a membrane.}
\end{figure}

Fig. 1 illustrates a solvated biomolecular system in an open domain $\Omega\in\mathbb{R}^{3}$. The open subdomain $\Omega_{m}\subset\Omega$ represents the biomolecule(s), and the remaining space $\Omega_{s} = \Omega\setminus\bar{\Omega}_{m}$ is filled with ionic solution
($s$ for solvent). Domains $\Omega_{m}$ and $\Omega_{s}$ are separated by a molecular surface $\Gamma_{m}$ (for simplicity, we call $\Gamma_m$ molecular surface in
the rest of the paper, but it also includes the membrane and nanopore surface if they exist).
The ionic flow can not penetrate the non-reactive molecular surface.
We use $\Gamma_D$ and $\Gamma_N$ to represent Dirichlet and Neumann boundary conditions, respectively.
According to the property of the physical system and model, both $\Gamma_D$ and $\Gamma_N$ can be applied to $\Gamma_s$ or part of $\Gamma_s$,
For examples, fixed potentials (Dirichlet BD) are usually given on the out boundary $\Gamma_s$ in PB calculations (Fig. 1(a))
and on the upper and lower boundaries of the whole box in PNP simulations (Fig. 1(b)). Surface charge density (Neumann BD)
is usually applied to the molecular/nanopore surface,\cite{Siwy06,Siwy07,Bzlu10} or a simplified molecular surface
(do not consider the molecular domain $\Omega_m$) \cite{LiBo16} to model the charge amount carried by the molecule.
The boundary of solvent region $\Gamma_{s} = \Gamma_{D} + \Gamma_{N}$.

Free energy discussions in previous works are usually for infinite domain with vanishing boundary conditions and do not
consider the non-zero Neumann and Dirichlet boundary effects. If we consider a finite or a confined region,
the variational approach to the derivation of free energy functional may face problems.
For an electrolyte solution system, the Gibbs free energy of the charged system is\cite{Sharp90,Gilson93,Fogolari97}
\begin{align}\label{2.1}
    F
    &= \frac{1}{2}\int_{\Omega}(\rho^{f}+\sum_{i}c_{i}q_{i})\phi dV
     + \int_{\Omega}\sum_{i}(kT\ln(c_{i}/c_{i}^{b})-kT)c_{i}dV\notag\\
    &= \int_{\Omega}\frac{1}{2}\rho\phi dV
     + \beta^{-1}\sum_{i=1}^{K}\int_{\Omega}c_{i}[\log(\Lambda^{3}c_{i})-1]dV - \sum_{i=1}^{K}\int_{\Omega}\mu_{i}c_{i}dV.
\end{align}
Here, $\rho$ is the total charge density, defined by
\begin{equation}\label{2.2}
    \rho = \rho^{f} + \sum_{i=1}^{K}q_{i}c_{i},
\end{equation}
where $q_{i} = Z_{i}e$ with $Z_{i}$ the valence of the $i$th ionic species and $e$ the elementary charge, $\rho^{f}$ is the permanent (fixed) charge distribution
\begin{equation}
    \rho^{f}(x) = \sum_{j}q_{j}\delta(x - x_{j}) \notag
\end{equation}
which is an ensemble of singular charges $q_{j}$ located at $x_{j}$ inside the biomolecule,
$\phi= \phi(c)$ is the electrostatic potential, $\beta^{-1} = k_{B}T$ with $k_{B}$ the Boltzmann constant and $T$ the temperature,
$\Lambda$ is the thermal de Broglie wavelength, $\mu_{i}$ is the chemical potential for the $i$ th ionic species, and $\mu_{i}^{b}$ is the standard-state chemical potential. The standard PB and PNP equations can be derived from variational method from this energy form.\cite{Bzlu08,WGWei12}

However, as aforementioned, in many real systems and/or numerical computing, the objective domain is finite, and people used to
adopt the same energy form and study different boundary conditions. This may lead to inconsistency among the energy form, PB/PNP equations
and the boundary conditions, and sometimes even resulted in nonphysical PDE model.
To obtain the consistent PDE(s), we need include different boundary interactions into the free energy functionals, and these new terms count for physical interactions with the boundary (for real boundary) or the environment influence on the computational domain system (for artificially
modeled boundary for numerical goal).
Generally, when there exists surface charges (denote the density as $\sigma$)
on the boundary or part of the boundary (where a Neumann boundary condition can be applied),
it is obvious to directly plug a surface energy term (${1 \over 2} \phi  \sigma$) into the free energy functional.
This is physically reasonable because
the surface charges cause an additional interaction with the electric field. This "improved" free energy is also often used and studied,
as in Ref.~\refcite{LiBo16}:
\begin{align}\label{2.3}
    F[c] &= \int_{\Omega}\frac{1}{2}\rho(c)\phi(c) dV + \int_{\Gamma_{N}}\frac{1}{2}\sigma\phi(c)dS \notag\\
    &+ \beta^{-1}\sum_{i=1}^{K}\int_{\Omega}c_{i}[\log(\Lambda^{3}c_{i})-1]dV - \sum_{i=1}^{K}\int_{\Omega}\mu_{i}c_{i}dV.
\end{align}
But this free energy is still not complete, as it lacks the treatment of Dirichlet boundary condition, which is rarely discussed in previous mathematical and physical work.
When a potential is given on a boundary, which means: (1) if the boundary is a physical boundary identified
as certain type of material interface, there must have a mount of surface charge to maintain the Dirichlet condition.
In physics, the surface charge density needs to be equal to $-\epsilon \frac{\partial \phi}{\partial n}$,
which thereby opposes an surface interaction energy $-{1\over 2} \epsilon \frac{\partial \phi}{\partial n}\phi$ to the total free energy; (2) if the boundary is an artificial
boundary (still immersed the electrolyte solution system), we are using a boundary condition to model the influence
from the "cutoff" outside part which is a polarizable dielectric media (environment). The influence can be approximated
by an "effective" surface charge as in the physical boundary case. This charge density also should be consistent with the
electric potential field and the given surface potential. In other words, the effective charge density is equal to ${1\over 2} \epsilon \frac{\partial \phi}{\partial n}$
and leads to a similar energy term.
Therefore, in either of above two cases, there also needs an energy term in the free energy functional for Dirichlet BD.
Here we present the complete free energy functional form:
\begin{align}\label{2.4}
    F[c] &= \int_{\Omega}\frac{1}{2}\rho(c)\phi(c) dV + \int_{\Gamma_{N}}\frac{1}{2}\sigma\phi(c)dS -\int_{\Gamma_{D}}\frac{1}{2}\epsilon(c)\frac{\partial \phi(c)}{\partial n}\phi_{0}dS \notag\\
     &+ \beta^{-1}\sum_{i=1}^{K}\int_{\Omega}c_{i}[\log(\Lambda^{3}c_{i})-1]dV - \sum_{i=1}^{K}\int_{\Omega}\mu_{i}c_{i}dV,
\end{align}
where $\phi= \phi(c)$ is the electrostatic potential determined as the solution to the general boundary-value problem of Poisson's equation
 \begin{align}\label{2.5}
  -\nabla\cdot(\epsilon(c)\nabla\phi(c)) &= \rho(c)     &in \;\Omega,\\
  \epsilon(c)\frac{\partial\phi}{\partial n} & = \sigma &on \;\Gamma_{N},\notag\\
  \phi & = \phi_{0} &on \;\Gamma_{D}.\notag
\end{align}
The first three terms in Eq.~\eqref{2.4} together represent the electrostatic potential energies, and in particular, the second and third terms are the boundary interactions. The fourth term represents the ideal-gas entropy and the last term in Eq.~\eqref{2.4}, represents the chemical potential of the system that results from the constraint of total number of ions in each species. It is worth noting that we here treat $\epsilon$ as a general inhomogeneous dielectric permittivity which is dependent on ionic concentration. This is another concerned topic of the paper.

In the next subsections, we will use the energetic variational approach to illustrate the correctness and consistency of above-mentioned free energy form. If the boundary interactions is missed in the free energy functionals, the energetic variational approach will produce some extra terms of boundary integration, and the Boltzmann distribution may not be obtained or be obtained in a screwed form. Of particular interest in the case of ionic concentration-dependent dielectric permittivity, the complete free energy form will correctly lead to two generalized equations under different boundary conditions.
%\newpage
\subsection{Energetic variational approach}
\subsubsection{First variations}
To derive the first variation of $F$ w.r.t. $c$, we first need the following basic assumptions:
\begin{romanlist}[(iii)]
\item The dielectric coefficient function $\epsilon(c) \in C^{1}([0,\infty))$. Moreover,
there are two positive numbers $\epsilon_{min}$ and $\epsilon_{max}$ such that
\begin{equation}\label{2.6}
    0 < \epsilon_{min} \leq \epsilon(c) \leq \epsilon_{max}\;\;\forall \;c\geq 0;
\end{equation}
\item \;$\Omega$ is bounded and open, $\Gamma = \partial\Omega = \Gamma_{N} + \Gamma_{D}$;
\item \;We also assume that a fixed charged density is given $\rho^{f} :\Omega\rightarrow R$, $\rho^{f}\in L^{\infty}(\Omega)$,
a surface charge density $\sigma : \Gamma_{N} \rightarrow R$,
and a boundary value of the electrostatic potential $\phi_{0} : \Gamma_{D} \rightarrow R$, $\phi_{0} \mid_{\Gamma_{D}}\in W^{2,\infty}(\Omega)$.
\end{romanlist}
We use the standard notion for Sobolev spaces:
$$H_{s}^{1} = \{\phi\in H^{1}(\Omega):\phi=\phi_{0} \quad on \;\Gamma_{D}\},$$
$$H_{s,0}^{1} = \{\phi\in H^{1}(\Omega):\phi=0 \quad on \;\Gamma_{D}\}.$$
The weak form of Eq.~\eqref{2.5} is
\begin{equation}
    \int_{\Omega}\nabla \cdot \epsilon(c)\nabla\phi(c) v dV = -\int_{\Omega}\rho(c)vdV \quad \forall v \in H_{s,0}^{1}(\Omega).\notag
\end{equation}
By the Gauss theorem, we have
\begin{equation}
    -\int_{\Omega}\epsilon(c)\nabla\phi(c) \nabla v dV
    + \int_{\Gamma}\epsilon(c)\frac{\partial \phi(c)}{\partial n}v dS
    = -\int_{\Omega}\rho(c)vdV \quad \forall v \in H_{s,0}^{1}(\Omega).\notag
\end{equation}
Let $\Gamma = \Gamma_{N} + \Gamma_{D}$, and $\forall v\mid_{\Gamma_{D}} = 0$, then we have
\begin{equation}\label{2.7}
    a(\phi, v) = \int_{\Omega}\epsilon(c)\nabla\phi(c) \nabla v dV = \int_{\Omega}\rho(c)vdV + \int_{\Gamma_{N}}\sigma v dS \quad \forall v \in H_{s,0}^{1}(\Omega).
\end{equation}
Since $L^{\infty}(\Omega)\cap H_{s,0}^{1}(\Omega)$ is dense in $H_{s,0}^{1}(\Omega)$, we can identify $u$ as an element in $H_{s,0}^{-1}(\Omega)$. We denote
\begin{align}
    X = \{c = (c_{1},\cdots,c_{K}) \in L^{1}(\Omega, R^{K}): c_{i} \geq 0 \quad a.e. &\quad\Omega, i=1, \cdots, K; \notag\\
    &\sum_{i=1}^{K}q_{i}c_{i}\in H_{s,0}^{-1}(\Omega)\}.\notag
\end{align}
Let $c \in X$, it follows from the Lax-Milgram theorem and the Poinc\'{a}re inequality for functions in $H_{s,0}^{1}(\Omega)$ that the boundary-value problem of Poisson equation Eq.~\eqref{2.5} has a unique weak solution $\phi = \phi(c)$.

Let $c = (c_{1},\cdots,c_{K}) \in X$ and $d = (d_{1},\cdots,d_{K}) \in X$, we define
\begin{equation}\label{2.8}
    \delta F[c][d] = \lim_{t\rightarrow 0}\frac{F[c+td] - F[c]}{t}.
\end{equation}
To get the expression of $\delta F[c][d]$, we need the following theorem.
\begin{theorem}
Let $c = (c_{1},\cdots,c_{K}) \in X$. Assume there exist positive numbers $\delta_{1}$ and $\delta_{2}$ such that
$\delta_{1} \leq c_{i}(x) \leq \delta_{2}$ for a.e. $x\in\Omega$ and $i=1, \cdots, K$. Assume also that $d = (d_{1},\cdots,d_{K}) \in L^{\infty}(\Omega, R^{K})$. Then
\begin{equation}\label{2.9}
    ||\phi(c+td)-\phi(c)||_{H^{1}(\Omega)}\rightarrow 0 \;\;as\;t\rightarrow 0.
\end{equation}
\end{theorem}
A proof of this theorem can be found in Ref.~\refcite{LiBo16}, and we will not repeat it here.

Now, we decompose the free energy $F$ as
\begin{equation*}
    F[c] = F_{pot}[c] + F_{entropy}[c],
\end{equation*}
where
\begin{equation}\label{2.10}
    F_{pot}[c] = \int_{\Omega}\frac{1}{2}\rho(c)\phi(c)dV + \int_{\Gamma_{N}} \frac{1}{2}\sigma\phi(c)dS
    - \int_{\Gamma_{D}} \frac{1}{2}\epsilon(c)\frac{\partial \phi(c)}{\partial n}\phi_{0}dS,
\end{equation}
\begin{equation}\label{2.11}
    F_{entropy}[c] = \sum_{i=1}^{K}\int_{\Omega}\{\beta^{-1}c_{i}[\log(\Lambda^{3}c_{i})-1]-\mu_{i}c_{i}\}dV.
\end{equation}
Based on the definition of \eqref{2.8}, we have
\begin{align}\label{2.12}
    \delta F_{entropy}[c][d] &= \lim_{t\rightarrow 0}\frac{F_{entropy}[c+td] - F_{entropy}[c]}{t}\notag\\
    & = \sum_{i=1}^{K}\int_{\Omega}d_{i}[\beta^{-1}\log(\Lambda^{3}c_{i})-\mu_{i}]dV.
\end{align}
We now deal with another term
\begin{align}\label{2.13}
    &\delta F_{pot}[c][d] = \lim_{t\rightarrow 0}\frac{F_{pot}[c+td] - F_{pot}[c]}{t}\notag\\
    & = \lim_{t\rightarrow 0}\frac{1}{2t}[\int_{\Omega}\{ \rho(c+td)\phi(c+td)-\rho(c)\phi(c+td)+\rho(c)\phi(c+td)-\rho(c)\phi(c)\}dV]\notag\\
    & + \lim_{t\rightarrow 0}\frac{1}{2t}\int_{\Gamma_{N}}\sigma[\phi(c+td)-\phi(c)]dS
      - \lim_{t\rightarrow 0}\frac{1}{2t}\int_{\Gamma_{D}}[\epsilon(c+td)\frac{\partial \phi(c+td)}{\partial n}-\epsilon(c)\frac{\partial \phi(c)}{\partial n}]\phi_{0}dS\notag\\
    & = \lim_{t\rightarrow 0}\frac{1}{2}\int_{\Omega}\frac{[\rho(c+td)-\rho(c)]\phi(c+td)}{t}dV + \lim_{t\rightarrow 0}\frac{1}{2}\int_{\Omega}\rho(c)\frac{\phi(c+td)-\phi(c)}{t}dV\notag\\
    & + \lim_{t\rightarrow 0}\frac{1}{2}\int_{\Gamma_{N}}\sigma\frac{\phi(c+td)-\phi(c)}{t}dV
      - \lim_{t\rightarrow 0}\frac{1}{2t}\int_{\Gamma_{D}}[\epsilon(c+td)\frac{\partial \phi(c+td)}{\partial n}-\epsilon(c)\frac{\partial \phi(c)}{\partial n}]\phi_{0}dS.
\end{align}
By Eq.~\eqref{2.2}, we have
\begin{equation*}
    \lim_{t\rightarrow 0}\frac{1}{2}\int_{\Omega}\frac{[\rho(c+td)-\rho(c)]\phi(c+td)}{t}dV
    = \lim_{t\rightarrow 0}\frac{1}{2}\sum_{i=1}^{K}\int_{\Omega}d_{i}q_{i}\phi(c+td)dV,
\end{equation*}
and by theorem 2.1, we have
\begin{equation}\label{2.14}
    \lim_{t\rightarrow 0}\frac{1}{2}\sum_{i=1}^{K}\int_{\Omega}q_{i}d_{i}\phi(c+td)dV
    = \frac{1}{2}\sum_{i=1}^{K}\int_{\Omega}d_{i}q_{i}\phi(c)dV.
\end{equation}
Now we deal with the remaining three terms in \eqref{2.13}, by the weak formulation \eqref{2.7} for $\phi(c)$
with $v = \frac{\phi(c+td)-\phi(c)}{t}\;\in \;H_{s,0}^{1}$,
\begin{align}\label{2.15}
    &\lim_{t\rightarrow 0}\frac{1}{2}\int_{\Omega}\rho(c)\frac{\phi(c+td)-\phi(c)}{t}dV
    + \lim_{t\rightarrow 0}\frac{1}{2}\int_{\Gamma_{N}}\sigma\frac{\phi(c+td)-\phi(c)}{t}dS\notag\\
    & - \lim_{t\rightarrow 0}\frac{1}{t}[\int_{\Gamma_{D}} \frac{1}{2}[\epsilon(c+td)\frac{\partial \phi(c+td)}{\partial n}-\epsilon(c)\frac{\partial \phi(c)}{\partial n}]\phi_{0}dS]\notag\\
    & = \lim_{t\rightarrow 0}\frac{1}{2}\int_{\Omega}\varepsilon(c)\nabla\phi(c)\nabla[\frac{\phi(c+td)-\phi(c)}{t}]dV\notag\\
    & - \lim_{t\rightarrow 0}[\int_{\Gamma_{D}} \frac{1}{2t}[\epsilon(c+td)\frac{\partial \phi(c+td)}{\partial n}-\epsilon(c)\frac{\partial \phi(c)}{\partial n}]\phi_{0}dS].
\end{align}
Based on the Poisson's Eq.~\eqref{2.5}, the following equation holds:
\begin{equation}\label{2.16}
    \int_{\Omega}-\nabla\cdot\epsilon(c)\nabla\phi(c)\phi(c)dV = \int_{\Omega}\rho(c)\phi(c)dV.
\end{equation}
By integrating the left term by parts and using the divergence theorem
\begin{equation}\label{2.17}
    \int_{\Omega}\epsilon(c)\nabla\phi(c)\nabla\phi(c)dV - \int_{\Gamma}\epsilon(c)\frac{\partial \phi(c)}{\partial n}\phi(c)dS= \int_{\Omega}\rho(c)\phi(c)dV.
\end{equation}
If we consider the Poisson's Eq.~\eqref{2.5} at $c+td$, similarly, we have
\begin{equation}\label{2.18}
    \int_{\Omega}\epsilon(c+td)\nabla\phi(c+td)\nabla\phi(c)dV - \int_{\Gamma}\epsilon(c+td)\frac{\partial \phi(c+td)}{\partial n}\phi(c)dS= \int_{\Omega}\rho(c+td)\phi(c)dV.
\end{equation}
If $\epsilon$ is constant, Eq.~\eqref{2.17} and Eq.~\eqref{2.18} lead to:
\begin{align*}
    &\int_{\Omega}\epsilon\nabla(\phi(c+td)-\phi(c))\nabla\phi(c)dV - \int_{\Gamma}\epsilon(\frac{\partial \phi(c+td)}{\partial n}-\frac{\partial \phi(c)}{\partial n})\phi(c)dS\\
    &= \int_{\Omega}(\rho(c+td)-\rho(c))\phi(c)dV.
\end{align*}
As the boundary $\Gamma$ of $\Omega$ is divided into two parts $\Gamma = \Gamma_{N} + \Gamma_{D}$, then
\begin{align*}
    &\int_{\Omega}\epsilon\nabla(\phi(c+td)-\phi(c))\nabla\phi(c)dV\\
    &- \int_{\Gamma_{N}}\epsilon(\frac{\partial \phi(c+td)}{\partial n}-\frac{\partial \phi(c)}{\partial n})\phi(c)dS
    - \int_{\Gamma_{D}}\epsilon(\frac{\partial \phi(c+td)}{\partial n}-\frac{\partial \phi(c)}{\partial n})\phi_{0}dS\\
    &= \int_{\Omega}\epsilon\nabla(\phi(c+td)-\phi(c))\nabla\phi(c)dV
    - \int_{\Gamma_{D}}\epsilon(\frac{\partial \phi(c+td)}{\partial n}-\frac{\partial \phi(c)}{\partial n})\phi_{0}dS\\
    &= \sum_{i=1}^{K}\int_{\Omega}td_{i}q_{i}\phi(c)dV.
\end{align*}
Take this equation into Eq.~\eqref{2.15}, then
\begin{align}\label{2.19}
    &\lim_{t\rightarrow 0}[\frac{1}{2}\int_{\Omega}\varepsilon\nabla\phi\nabla[\frac{\phi(c+td)-\phi(c)}{t}]dV]
    - \lim_{t\rightarrow 0}[\int_{\Gamma_{D}} \frac{1}{2t}[\epsilon\frac{\partial \phi(c+td)}{\partial n}-\epsilon\frac{\partial \phi(c)}{\partial n}]\phi_{0}dS]\notag\\
    &= \frac{1}{2}\sum_{i=1}^{K}\int_{\Omega}d_{i}q_{i}\phi(c)dV.
\end{align}
Combine Eqs.~\eqref{2.12},~\eqref{2.14} and \eqref{2.19}, when $\epsilon$ is constant we finally have
\begin{align*}
    &\delta F[c][d]  = \delta F_{entropy}[c][d] + \delta F_{pot}[c][d]\\
    & = \sum_{i=1}^{K}\int_{\Omega}d_{i}[\beta^{-1}\log(\Lambda^{3}c_{i})-\mu_{i}]dV
    + \frac{1}{2}\sum_{i=1}^{K}\int_{\Omega}d_{i}q_{i}\phi(c)dV
    + \frac{1}{2}\sum_{i=1}^{K}\int_{\Omega}d_{i}q_{i}\phi(c)dV\\
    & = \sum_{i=1}^{K}\int_{\Omega}d_{i}\{q_{i}\phi(c) + \beta^{-1}\log(\Lambda^{3}c_{i}) - \mu_{i}\}dV.\;
\end{align*}
If $\epsilon(c)$ is a function of $c$, we can deduce the equation below from Eq.~\eqref{2.17} and Eq.~\eqref{2.18}:
\begin{align*}
    &\int_{\Omega}[\epsilon(c+td)\nabla\phi(c+td)-\epsilon(c)\nabla\phi(c)]\nabla\phi(c)dV
    - \int_{\Gamma}[\epsilon(c+td)\frac{\partial \phi(c+td)}{\partial n}-\epsilon(c)\frac{\partial \phi(c)}{\partial n}]\phi(c)dS\\
    &= \int_{\Omega}[(\epsilon(c+td)-\epsilon(c))\nabla\phi(c+td)+\epsilon(c)(\nabla\phi(c+td)-\nabla\phi(c))]\nabla\phi(c)dV\\
    &- \int_{\Gamma_N+\Gamma_D}[\epsilon(c+td)\frac{\partial \phi(c+td)}{\partial n}-\epsilon(c)\frac{\partial \phi(c)}{\partial n}]\phi(c)dS\\
    &= \sum_{i=1}^{K}\int_{\Omega}[(td_{i}\epsilon^{'}(c) + o(t))\nabla\phi(c+td)\nabla\phi(c)]dV
    + \int_{\Omega}[\epsilon(c)(\nabla\phi(c+td)-\nabla\phi(c))]\nabla\phi(c)dV\\
    & -\int_{\Gamma_D}[\epsilon(c+td)\frac{\partial \phi(c+td)}{\partial n}-\epsilon(c)\frac{\partial \phi(c)}{\partial n}]\phi_0dS\\
    &= \sum_{i=1}^{K}\int_{\Omega}td_{i}q_{i}\phi(c)dV,
\end{align*}
 where we denote $\epsilon^{'}(c)$ as $\frac{\partial \epsilon(c)}{\partial c_i}$, and take above equation into Eq.~\eqref{2.15}
\begin{align}\label{2.20}
    &\lim_{t\rightarrow 0}\frac{1}{2}\int_{\Omega}\epsilon(c)\nabla\phi\nabla[\frac{\phi(c+td)-\phi(c)}{t}]dV
    - \lim_{t\rightarrow 0}[\int_{\Gamma_{D}} \frac{1}{2t}[\epsilon(c+td)\frac{\partial \phi(c+td)}{\partial n}-\epsilon(c)\frac{\partial \phi(c)}{\partial n}]\phi_{0}dS]\notag\\
    &=- \lim_{t\rightarrow 0}\int_{\Omega}\frac{1}{2}\sum_{i=1}^{K}d_{i}\epsilon^{'}(c)\nabla\phi(c+td)\nabla\phi(c)dV
     + \frac{1}{2}\sum_{i=1}^{K}\int_{\Omega}d_{i}q_{i}\phi(c)dV\notag\\
    &+ \lim_{t\rightarrow 0}\frac{1}{2t}\int_{\Gamma_{D}}[\epsilon(c+td)\frac{\partial \phi(c+td)}{\partial n}-\epsilon(c)\frac{\partial \phi(c)}{\partial n}]\phi(c)dS\notag\\
    &- \lim_{t\rightarrow 0}[\int_{\Gamma_{D}} \frac{1}{2t}[\epsilon(c+td)\frac{\partial \phi(c+td)}{\partial n}-\epsilon(c)\frac{\partial \phi(c)}{\partial n}]\phi_{0}dS]\notag\\
    &= \frac{1}{2}\sum_{i=1}^{K}\int_{\Omega}d_{i}q_{i}\phi(c)dV
     - \frac{1}{2}\sum_{i=1}^{K}\int_{\Omega}d_{i}\epsilon^{'}(c)\nabla\phi(c)\nabla\phi(c)dV.
\end{align}
Combine Eqs.~\eqref{2.12},~\eqref{2.14} and \eqref{2.20}, we finally have
\begin{align*}
    \delta F[c][d] & = \delta F_{entropy}[c][d] + \delta F_{pot}[c][d]\\
    & = \sum_{i=1}^{K}\int_{\Omega}d_{i}[\beta^{-1}\log(\Lambda^{3}c_{i})-\mu_{i}]dV
    + \frac{1}{2}\sum_{i=1}^{K}\int_{\Omega}q_{i}d_{i}\phi(c)dV
    + \frac{1}{2}\sum_{i=1}^{K}\int_{\Omega}d_{i}q_{i}\phi(c)dV\\
    &-\int_{\Omega}\frac{1}{2}\sum_{i=1}^{K}d_{i}\epsilon^{'}(c)\nabla\phi(c)\nabla\phi(c)dV\\
    & = \sum_{i=1}^{K}\int_{\Omega}d_{i}\{q_{i}\phi(c) + \beta^{-1}\log(\Lambda^{3}c_{i}) - \mu_{i}-\frac{1}{2}\epsilon^{'}(c)\nabla\phi(c)\nabla\phi(c)\}dV.\;
\end{align*}
In the case of inhomogeneous dielectric coefficient based on these discussions, we can prove the following theorem.
\begin{theorem}
Let $c = (c_{1},\cdots,c_{K}) \in X$. Assume there exist positive numbers $\delta_{1}$ and $\delta_{2}$ such that
$\delta_{1} \leq c_{i}(x) \leq \delta_{2}$ for a.e. $x\in\Omega$ and $i=1, \cdots, K$. Assume also that $d = (d_{1},\cdots,d_{K}) \in L^{\infty}(\Omega, R^{K})$. If we consider the complete free energy functional as given in Eq.~\eqref{2.4}, then
\begin{equation}\label{2.21}
    \delta F[c][d] = \sum_{i=1}^{K}\int_{\Omega}d_{i}\{q_{i}\phi(c)- \frac{1}{2}\epsilon^{'}(c)\nabla\phi(c)\nabla\phi(c) + \beta^{-1}\log(\Lambda^{3}c_{i}) - \mu_{i}\}dV.
\end{equation}
Particularly, if $\epsilon$ doesn't depend on $c$, then
\begin{equation}\label{2.22}
    \delta F[c][d] = \sum_{i=1}^{K}\int_{\Omega}d_{i}\{q_{i}\phi(c) + \beta^{-1}\log(\Lambda^{3}c_{i}) - \mu_{i}\}dV.
\end{equation}
\end{theorem}
\subsubsection{Comparison with result from the incomplete energy form}
To compare with result from the incomplete energy form, we use the energetic variational approach to
the incomplete free energy functional \eqref{2.3} rather than \eqref{2.4} in a finite domain (or semi-finite domain as well), and theoretical analysis will give essentially different results. An extra surface integral occurs in the first variations $\delta F[c][d]$ despite of the dependency of the dielectric coefficient on ionic concentrations:
\begin{align}\label{2.23}
    \delta F[c][d]
    & = \sum_{i=1}^{K}\int_{\Omega}d_{i}\{q_{i}\phi(c) + \beta^{-1}\log(\Lambda^{3}c_{i}) - \mu_{i}-\frac{1}{2}\epsilon^{'}(c)\nabla\phi(c)\nabla\phi(c)\}dV\notag\\
    &+\lim_{t\rightarrow 0}\frac{1}{2t}\int_{\Gamma_{D}}[\epsilon(c+td)\frac{\partial \phi(c+td)}{\partial n}-\epsilon(c)\frac{\partial \phi(c)}{\partial n}]\phi_{0}dS.\;
\end{align}
The boundary integration term is introduced by the non-zero Dirichlet boundary condition. A general method to eliminate this effect is to introduce a corresponding boundary-value problem of Poisson's equation as shown in Li et al.'s work \cite{LiBo16}
\begin{align}\label{2.24}
  \nabla\cdot(\epsilon(c)\nabla\phi_{D}(c)) &= 0     & in \quad \Omega,\\
  \epsilon(c)\frac{\partial\phi_{D}}{\partial n} & = 0& on \quad \Gamma_{N},\notag\\
  \phi_{D} & = \phi_{0} & on \quad \Gamma_{D}.\notag
\end{align}
Similarly, the weak form of Eq.~\eqref{2.24} is
\begin{equation}\label{2.25}
    \int_{\Omega} \epsilon(c)\nabla\phi_{D}(c) \cdot\nabla v dV = 0 \quad \forall v \in H_{s,0}^{1}(\Omega).
\end{equation}
The boundary-value problem of Poisson equation Eq.~\eqref{2.24} has a unique weak solution $\phi_{D} = \phi_{D}(c)$ and only in the special case of zero boundary condition $\phi_{0}=0$, the introduced $\phi_D$ vanishes $\phi_{D}=0$.
\begin{theorem}\cite{LiBo09a,LiBo09b,LiBo16}
Let $c = (c_{1},\cdots,c_{K}) \in X$. Assume there exists positive numbers $\delta_{1}$ and $\delta_{2}$ such that
$\delta_{1} \leq c_{i}(x) \leq \delta_{2}$ for a.e. $x\in\Omega$ and $i=1, \cdots, K$. Assume also that $d = (d_{1},\cdots,d_{K}) \in L^{\infty}(\Omega, R^{K})$. If we consider the incomplete free energy functional as \eqref{2.3}, then $$\delta F[c][d] = \sum_{i=1}^{K}\int_{\Omega}d_{i}\delta_{i}F[c]dV,$$
where for each $i(1 \leq i \leq K)$ the function $\delta_{i}F[c]: \Omega \rightarrow R$ is given by:
\begin{equation}\label{2.26}
    \delta_{i}F[c] = q_{i}[\phi(c)-\frac{1}{2}\phi_{D}(c)]- \frac{1}{2}\epsilon^{'}(c)\nabla\phi(c)\cdot \nabla[\phi(c)-\phi_{D}(c)] + \beta^{-1}\log(\Lambda^{3}c_{i}) - \mu_{i}.
\end{equation}
\end{theorem}
\begin{proof}
Based on Eq.~\eqref{2.15} and by the weak formulation in Eq.~\eqref{2.7} for $\phi(c+td)$ and $\phi(c)$, and the weak formulation in Eq.~\eqref{2.25} for $\phi_{D}$
with $v = \frac{\phi(c+td)-\phi(c)}{t}\;\in \;H_{s,0}^{1}$ and $v = \phi(c)-\phi_{D}(c)\;\in \;H_{s,0}^{1}$,
\begin{align}\label{2.27}
    &\lim_{t\rightarrow 0}\frac{1}{2}\int_{\Omega}\epsilon(c)\nabla\phi(c)\nabla[\frac{\phi(c+td)-\phi(c)}{t}]dV\notag\\
    & = \lim_{t\rightarrow 0}\frac{1}{2}\int_{\Omega}\epsilon(c)\nabla[\phi(c)-\phi_{D}(c)]\nabla[\frac{\phi(c+td)-\phi(c)}{t}]dV\notag\\
    & = \lim_{t\rightarrow 0}[\frac{1}{2t}\int_{\Omega}(\epsilon(c)-\epsilon(c+td))\nabla[\phi(c)-\phi_{D}(c)]\nabla\phi(c+td)dV]\notag\\
    & + \lim_{t\rightarrow 0}[\frac{1}{2t}\{\rho(c+td)[\phi(c)-\phi_{D}(c)] dV
    + \int_{\Gamma_{D}}\epsilon(c+td)\frac{\partial\phi(c+td)}{\partial n}[\phi(c)-\phi_{D}(c)] dS\}]\notag\\
    & - \lim_{t\rightarrow 0}[\frac{1}{2t}\{\rho(c)[\phi(c)-\phi_{D}(c)] dV
    + \int_{\Gamma_{D}}\epsilon(c)\frac{\partial\phi}{\partial n}[\phi(c)-\phi_{D}(c)] dS\}]\notag\\
    & = -\frac{1}{2}\sum_{i=1}^{K}\int_{\Omega}d_{i}\epsilon^{'}(c)\nabla[\phi(c)-\phi_{D}(c)]\nabla\phi(c) dV
    + \frac{1}{2}\sum_{i=1}^{K}\int_{\Omega}d_{i}q_{i}[\phi(c)-\phi_{D}(c)] dV.
\end{align}
Combine Eqs.~\eqref{2.12},~\eqref{2.14} and \eqref{2.27}, we have
\begin{align*}
    \delta F[c][d] & = \delta F_{entropy}[c][d] + \delta F_{pot}[c][d]\\
    & = \sum_{i=1}^{K}\int_{\Omega}d_{i}\{q_{i}[\phi(c)-\frac{1}{2}\phi_{D}(c)] + \beta^{-1}\log(\Lambda^{3}c_{i}) - \mu_{i}\}dV\\
    & - \sum_{i=1}^{K}\int_{\Omega}d_{i}\frac{1}{2}\epsilon^{'}(c)\nabla\phi(c)\cdot \nabla[\phi(c)-\phi_{D}(c)]dV.\;
\end{align*}
\end{proof}
This will lead to "screwed" PB and PNP models and obtain incorrect results in physics. In the next two subsections, we will derive the generalized PB/PNP equations and give detailed discussion.
\subsection{Generalized boltzmann distributions with different boundary conditions}
Based on the complete free energy functional \eqref{2.4} and theorem 2.2, the electrostatic free energy $F = F(c)$ is minimized when $c = (c_{1},\cdots,c_{K}) \in X$ satisfies
$\delta F[c][d] = 0, \quad \forall d = (d_{1},\cdots,d_{K}) \in X$, which means
\begin{equation*}
    q_{i}\phi(c)- \frac{1}{2}\epsilon^{'}(c)\nabla\phi(c)\nabla\phi(c) + \beta^{-1}\log(\Lambda^{3}c_{i}) - \mu_{i} = 0.
\end{equation*}
$\Rightarrow$
\begin{align}\label{2.28}
    c_{i} &= \Lambda^{-3}e^{\beta\mu_{i}}\exp\{-\beta q_{i}\phi(c)+ \frac{\beta}{2}\epsilon^{'}(c)\nabla\phi(c) \nabla\phi(c)\}\notag\\
          &= c_{i}^{\infty}\exp\{-\beta q_{i}\phi(c)+ \frac{\beta}{2}\epsilon^{'}(c)\nabla\phi(c) \nabla\phi(c)\},
\end{align}
where $c_{i} \rightarrow c_{i}^{\infty}$ as $r \rightarrow \infty$ and $\phi \rightarrow 0$.
We call these the generalized Boltzmann distributions, as they generalize the classical Boltzmann distributions
$c_{i} = c_{i}^{\infty}e^{-\beta q_{i}\phi}(i=1, \cdots, K)$ when $\epsilon$ does not depend on $c$ (no matter what the boundary conditions are).

However, if we start from the incomplete free energy functional \eqref{2.1} in a finite domain (or similarly for semi-finite domain) with non-zero Neumann/Dirichlet boundary conditions, $\delta F[c][d]$ takes the form,
\begin{align*}
    \delta F[c][d]  &= \sum_{i=1}^{K}\int_{\Omega}d_{i}\{q_{i}\phi(c) + \beta^{-1}\log(\Lambda^{3}c_{i}) - \mu_{i}\}dV\\
    & - \lim_{t\rightarrow 0}[\int_{\Gamma_{N}} \frac{1}{2}\sigma\frac{\phi(c+td)-\phi(c)}{t}dS]\\
    & + \lim_{t\rightarrow 0}[\int_{\Gamma_{D}} \frac{1}{2t}[\epsilon\frac{\partial \phi(c+td)}{\partial n}-\epsilon\frac{\partial \phi(c)}{\partial n}]\phi_{0}dS].
\end{align*}
Then we cannot obtain a generalized Boltzmann distribution.
Based on theorem 2.3 and minimize the incomplete energy functional \eqref{2.3}, a screwed Boltzmann distribution can be derived.

Here we give an example to quantify the difference of these two distributions. If $\epsilon$ does not depend on $c$, the generalized Boltzmann distributions \eqref{2.28} are exactly the same as the classical Boltzmann distributions
\begin{equation*}
    c_{i} = c_{i}^{\infty}e^{-\beta q_{i}\phi},
\end{equation*}
and the "screwed" (non-physical) Boltzmann distributions take the form,
\begin{equation}\label{2.29}
    c_{i} = c_{i}^{\infty}\exp\{-\beta q_{i}(\phi(c)-\frac{1}{2}\phi_{D}(c))\}.
\end{equation}
In this example, we design a virtual (ideal) numerical experiment. Considering a charged sphere in an infinite ionic solution, the bulk concentration ($r \rightarrow \infty$) is $c_{i}^{\infty} = 0.1M$ and when $r \rightarrow \infty$, $\phi \rightarrow 0$. In numerical calculation,
the computational domain is finite, we set $\phi = \phi_{D}$ as the Dirichlet boundary condition on an imaginary
 spherical boundary at distance $r = R$. Supposing $\phi_D$ is the real value (depending on the charged sphere and ionic strength)
 of the real system, the numerical solution should match the realistic potential and concentration distributions.
But apparently at $r = R$ (at the boundary) the above two Boltzmann distributions lead to discrepancy in concentration predictions,
one is $c_{i}^{\infty}e^{-\beta q_{i}\phi_D}$, one is $c_{i}^{\infty}e^{-{1 \over 2}\beta q_{i}\phi_D}$. Fig. 2 draws the difference as
a function of $\phi_D$. It is notable that the gap between the two concentration predictions at the boundary becomes larger with the increase of applied potentials. When the fixed potential is a positive, the "screwed" Boltzmann distributions lead to lower concentrations for anions, and higher concentration for cations.
For negative boundary potential $\phi_D$, the opposite phenomenon occurs. When the fixed potential is zero, the distributions reduce to the same Boltzmann distribution.
\begin{figure}[pb]
\centerline{\psfig{file=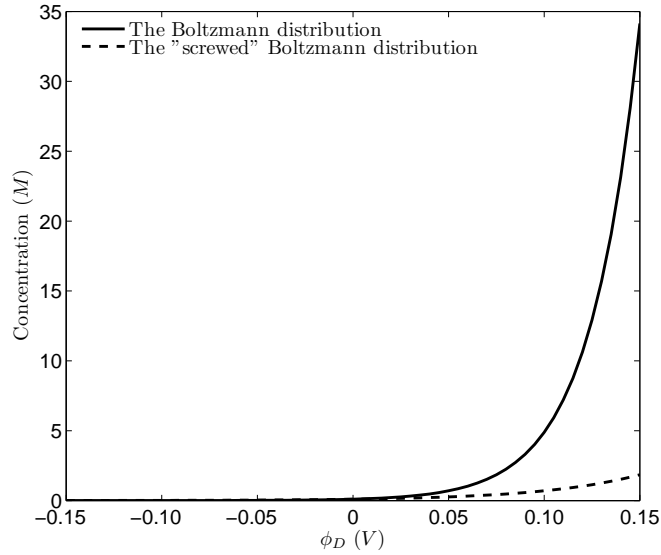,width=4in}}
\vspace*{8pt}
\caption{The traditional (solid line) and "screwed" (dashed line) Boltzmann distributions at the Dirichlet boundary as a function of the boundary value $\phi_D$ and bulk concentration 0.1M.}
\end{figure}
An alternative example can also be designed as a "semi-open" electrolyte solution system which has a Dirichlet BD ($\phi=\phi_D$) at a "finite" part of the boundary, and has a zero boundary condition at the infinite boundary  ($\phi \rightarrow 0$, $c_{i} \rightarrow c_{i}^{\infty}$ as $r \rightarrow \infty$). Similarly as above example, on the finite boundary ($\phi = \phi_{D}$), the generalized Boltzmann distribution is exactly the classical Boltzmann distributions $c_{i} = c_{i}^{\infty}e^{-\beta q_{i}\phi_{D}}$, while the screwed Boltzmann distributions $c_{i} = c_{i}^{\infty}e^{-\frac{1}{2}\beta q_{i}\phi_{D}}$ lead to wrong results.
\subsection{Generalized PNP equations with concentration-dependent $\epsilon(c)$ and different boundary conditions}
Ionic diffusion in electrolyte solution is an electro-diffusion process that is influenced by the electric field generated by the ion distribution itself, biomolecule(s) (if existed) and the environment.
The Poisson-Nernst-Planck equations coupling the electric potential and ion
concentration distributions provide an ideal model for describing this process.\cite{Eisenberg93,Bzlu10} The PNP equations have been widely
used to study the ion channels, nanopores, fuel cells and other research areas.\cite{Eisenberg93,Bob98,Cardenas00,Dan09,Bzlu11,Siwy07,WGWei12}
The continuum PNP equations can be derived via different routes. They can be obtained from the microscopic model of
Langevin trajectories in the limit of large damping and absence of correlations of different ionic trajectories,
\cite{Schuss01,Nadler04} or from the variations of the free energy functional that includes the electrostatic
free energy and the ideal component of the chemical potential (Eq.~\eqref{1.2}).\cite{Dirk02} As aforementioned, the previous variational method can only ensure consistency between the energy form and the PNP equations for vanished boundary conditions for electric potential $\phi$ such as for infinite domain because they did not include the boundary interaction terms. In addition, an inhomogeneously concentration-dependent dielectric property causes wide research interest recently.\cite{Hasted48,HO78,Nortemann97,Berk06,Hlli14,LiBo16} But little previous study is found to give a consistent dynamic model (such as PNP) for electrolyte solution when the dielectric coefficient is ionic concentration-dependent. This is also to be studied in current subsection.

We start from the free energy functional given by Eq.~\eqref{2.4} with generic Dirichlet and Neumann BDs.
According to the constitutive relations, the flux $J_{i}$ and the electrochemical potential $\mu_{i}$ of the $i$th species satisfy
\begin{equation*}
    J_{i} = -m_{i}c_{i}\nabla\mu_{i},
\end{equation*}
here $m_{i}$ is the ion mobility that relates to its diffusivity $D_{i}$ through Einsteins
relation $D_{i} = \beta^{-1}m_{i}$, $\mu_{i}$ is the variation of $F$ with respect to $c_{i}$:
\begin{equation*}
    \mu_{i} = \frac{\delta F}{\delta c_{i}}.
\end{equation*}
Then the following transport equations are obtained from the mass and current conservation law:
\begin{align*}
    \frac{\partial c_{i}}{\partial t}
&= -\nabla \cdot J_{i}\\
&= \nabla \cdot (\beta D_{i}c_{i}\nabla\mu_{i})\\
&=\nabla \cdot (\beta D_{i}c_{i}\nabla\{q_{i}\phi(c)- \frac{1}{2}\epsilon^{'}(c)\nabla\phi(c)\cdot \nabla\phi(c) + \beta^{-1}\log(\Lambda^{3}c_{i})\})\\
&=\nabla \cdot (\beta D_{i}c_{i}(\frac{\nabla c_{i}}{\beta c_{i}}+\nabla(q_{i}\phi(c)- \frac{1}{2}\epsilon^{'}(c)\nabla\phi(c)\cdot \nabla\phi(c))))\\
&=\nabla \cdot (D_{i}[\nabla c_{i}+\beta c_{i}\nabla(q_{i}\phi(c)- \frac{1}{2}\epsilon^{'}(c)\nabla\phi(c)\cdot \nabla\phi(c))]).
\end{align*}
Now we get a set of generalized self-consistent PNP equations:
\begin{equation}\label{2.30}
    -\nabla\cdot\epsilon(c)\nabla\phi(c) = \rho^{f} + \lambda\sum_{i=1}^{K}q_{i}c_{i}, \quad in \;\Omega,
\end{equation}
\begin{equation}\label{2.31}
    \frac{\partial c_{i}}{\partial t}=\nabla \cdot (D_{i}[\nabla c_{i}+\beta c_{i}\nabla(q_{i}\phi(c)- \frac{1}{2}\epsilon^{'}(c)\nabla\phi(c)\cdot \nabla\phi(c))]), \quad in \;\Omega_{s}, i=1, 2, \cdots, K.
\end{equation}
\begin{align*}
\epsilon(c)\frac{\partial\phi}{\partial n} &= \sigma \quad &on \;\Gamma_{N},\\
\phi &= \phi_{0} \quad &on \;\Gamma_{D},\\
c_{i} &= c_{i}^{b} \quad &on \;\Gamma_{D},\\
J_{i}\cdot n &= 0 \quad &on \;\Gamma_{m}.
\end{align*}
If the dielectric coefficient does not depend on local ionic concentrations, Eqs.~\eqref{2.30}-\eqref{2.31} will reduce to the traditional PNP equations.
\begin{equation}\label{2.32}
    -\nabla\cdot\epsilon\nabla\phi(c) = \rho^{f} + \lambda\sum_{i=1}^{K}q_{i}c_{i}, \quad in \;\Omega,
\end{equation}
\begin{equation}\label{2.33}
    \frac{\partial c_{i}}{\partial t}=\nabla \cdot (D_{i}[\nabla c_{i}+\beta c_{i}\nabla(q_{i}\phi(c))]),
    \quad in \;\Omega_{s}, i=1, 2, \cdots, K.
\end{equation}
\begin{align*}
\epsilon\frac{\partial\phi}{\partial n} &= \sigma \quad &on \;\Gamma_{N},\\
\phi &= \phi_{0} \quad &on \;\Gamma_{D},\\
c_{i} &= c_{i}^{b} \quad &on \;\Gamma_{D},\\
J_{i}\cdot n &= 0 \quad &on \;\Gamma_{m}.
\end{align*}
However, for simplicity, if $\epsilon$ does not depend on $c$, but $\phi_{0} \neq 0$, according to theorem 2.3, the PNP equations from the incomplete energy form \eqref{2.3} take the form of
\begin{equation}\label{2.34}
    -\nabla\cdot\epsilon\nabla\phi(c) = \rho^{f} + \lambda\sum_{i=1}^{K}q_{i}c_{i}, \quad in \;\Omega,
\end{equation}
\begin{equation}\label{2.35}
    \frac{\partial c_{i}}{\partial t}=\nabla \cdot (D_{i}[\nabla c_{i}+\beta c_{i}\nabla(q_{i}[\phi(c)-\frac{1}{2}\phi_{D}(c)])]),
    \quad in \;\Omega_{s}, i=1, 2, \cdots, K.
\end{equation}
\begin{align*}
\epsilon\frac{\partial\phi}{\partial n} &= \sigma \quad &on \;\Gamma_{N},\\
\phi &= \phi_{0} \quad &on \;\Gamma_{D},\\
c_{i} &= c_{i}^{b} \quad &on \;\Gamma_{D},\\
J_{i}\cdot n &= 0 \quad &on \;\Gamma_{m}.
\end{align*}
Obviously, this is inconsistent with the established physics in this area. The drift term in the right hand side of Eq.~\eqref{2.35} originates from the electric field driving ($\nabla\phi$) and should be irrelevant to $\phi_{D}$ which is introduced only for mathematical analysis of the incomplete free energy form and shouldn't change the physical phenomenon. The physical phenomenon should not changed by $\phi_D$. Therefore, this is actually another main reason to question the previous energy functionals. It also suggests that adding the boundary interactions into the free energy is necessary to make it consistent to PDEs. In subsection 2.4.2,
we will give numerical simulations for a cylinder nanopore to further study the different current-voltage predictions from these two derived new PNP models.
In next subsection we will prove that the complete energy functional form \eqref{2.4} satisfies the energy dissipation law.
\subsubsection{Energy dissipation law}
Electro-diffusion process in electrolyte solution is a energy dissipation process. This requires that the evolutionary equation system such as the PNP equations need to satisfy the energy dissipation law. This subsection studies the energy dissipation properties of the energy forms and the PNP systems.
We first consider the free energy functional \eqref{2.1} with isothermal assumption and vanishing boundary conditions. For simplicity, a constant $\epsilon$ is considered, the ionic system \eqref{2.32}-\eqref{2.33} has been shown in Ref.~\refcite{Shixin13} to satisfy the following energy dissipation law,
\begin{align*}
  \frac{d}{dt}E^{total} &=\frac{d}{dt}
  [\int_{\Omega}(K_{B}T(c_{1}\ln\frac{c_{1}}{c_{1}^{\infty}}+ c_{2}\ln\frac{c_{2}}{c_{2}^{\infty}} )
  +\frac{\epsilon}{2}|\nabla \phi|^{2})dx] \\
     &= -\int_{\Omega}[\frac{D_{1}}{K_{B}T}c_{1}|\nabla \mu_{1}|^{2}
    +\frac{D_{2}}{K_{B}T}c_{2}|\nabla \mu_{2}|^{2}]dx.
\end{align*}
If the PNP system \eqref{2.32}-\eqref{2.33} with generic Dirichlet and Neumann boundary conditions on the outer boundary, it satisfies the energy law,
\begin{align*}
  \frac{d}{dt}E^{total} &=\frac{d}{dt}
  [\int_{\Omega}(K_{B}T(c_{1}\ln\frac{c_{1}}{c_{1}^{\infty}}+ c_{2}\ln\frac{c_{2}}{c_{2}^{\infty}} )
  +\frac{\epsilon}{2}|\nabla \phi|^{2})dx] \\
     &= -\int_{\Omega}[\frac{D_{1}}{K_{B}T}c_{1}|\nabla \mu_{1}|^{2}
    +\frac{D_{2}}{K_{B}T}c_{2}|\nabla \mu_{2}|^{2}]dx\\
    &+\int_{\Gamma_{D}}\epsilon\phi_{0}\frac{d}{dt}(\frac{\partial\phi}{\partial n})dx
    +\int_{\Gamma_{N}}\frac{d}{dt}(\sigma\phi)dx.
\end{align*}
If the last two terms are large enough, this PNP system doesn't satisfy the energy dissipation law.
But using the complete free energy functional \eqref{2.4}, we will show as following that
if we begin with the complete free energy functional \eqref{2.4}, the aforementioned PNP system also satisfies the energy dissipation law,
\begin{align*}
  \frac{d}{dt}E^{total}
  &=\frac{d}{dt}
  [\int_{\Omega}(K_{B}T(c_{1}\ln\frac{c_{1}}{c_{1}^{\infty}}+ c_{2}\ln\frac{c_{2}}{c_{2}^{\infty}} )
  +\frac{1}{2}\rho\phi)dx - \int_{\Gamma_{N}}\frac{1}{2}\sigma\phi dx
  -\int_{\Gamma_{D}}\frac{1}{2}\epsilon\frac{\partial \phi}{\partial n}\phi_{0}dx] \\
  &=\frac{d}{dt}
  [\int_{\Omega}(K_{B}T(c_{1}\ln\frac{c_{1}}{c_{1}^{\infty}}+ c_{2}\ln\frac{c_{2}}{c_{2}^{\infty}} )
  +\frac{\epsilon}{2}|\nabla \phi|^{2})dx -\int_{\Gamma_{N}}\sigma\phi dx
  - \int_{\Gamma_{D}}\epsilon\frac{\partial \phi}{\partial n}\phi_{0}dx] \\
  &= -\int_{\Omega}[\frac{D_{1}}{K_{B}T}c_{1}|\nabla \mu_{1}|^{2}
  +\frac{D_{2}}{K_{B}T}c_{2}|\nabla \mu_{2}|^{2}]dx.
\end{align*}
The dissipation functional is a sum of two parts, which are all non-positive. This indicates that the "true" total energy defined in \eqref{2.4} do decrease along with the dissipative electro-diffusion process.
\subsubsection{Numerical simulation in a cylinder nanopore system}
In this subsection, we present an example with a cylinder nanopore to further investigate the difference between the standard traditional PNP and the "screwed" PNP models.
A cylinder nanopore with a height of 50$\AA$ and a pore radius of 2$\AA$ is placed in the middle of a cubic box of
$100\AA \times 100\AA \times 100\AA $.
A charge density is $-0.02C/m^{2}$ is set on the inner surface of the nanopore and the potential on the lower boundary of the cubic box is fixed to be zero, while the upper boundary values (taken as membrane potentials) change from -200mV to 200mV with a step length of 50mV.
In this example, we use a finite element method to solve these PNP equations in the solvent region $\Omega_s$ and do not consider the molecular domain $\Omega_m$. The geometry and a mesh of the cylinder nanopore is illustrated in Fig. 3.
\begin{figure}[pb]
\centering
  \includegraphics[width=1.9in]{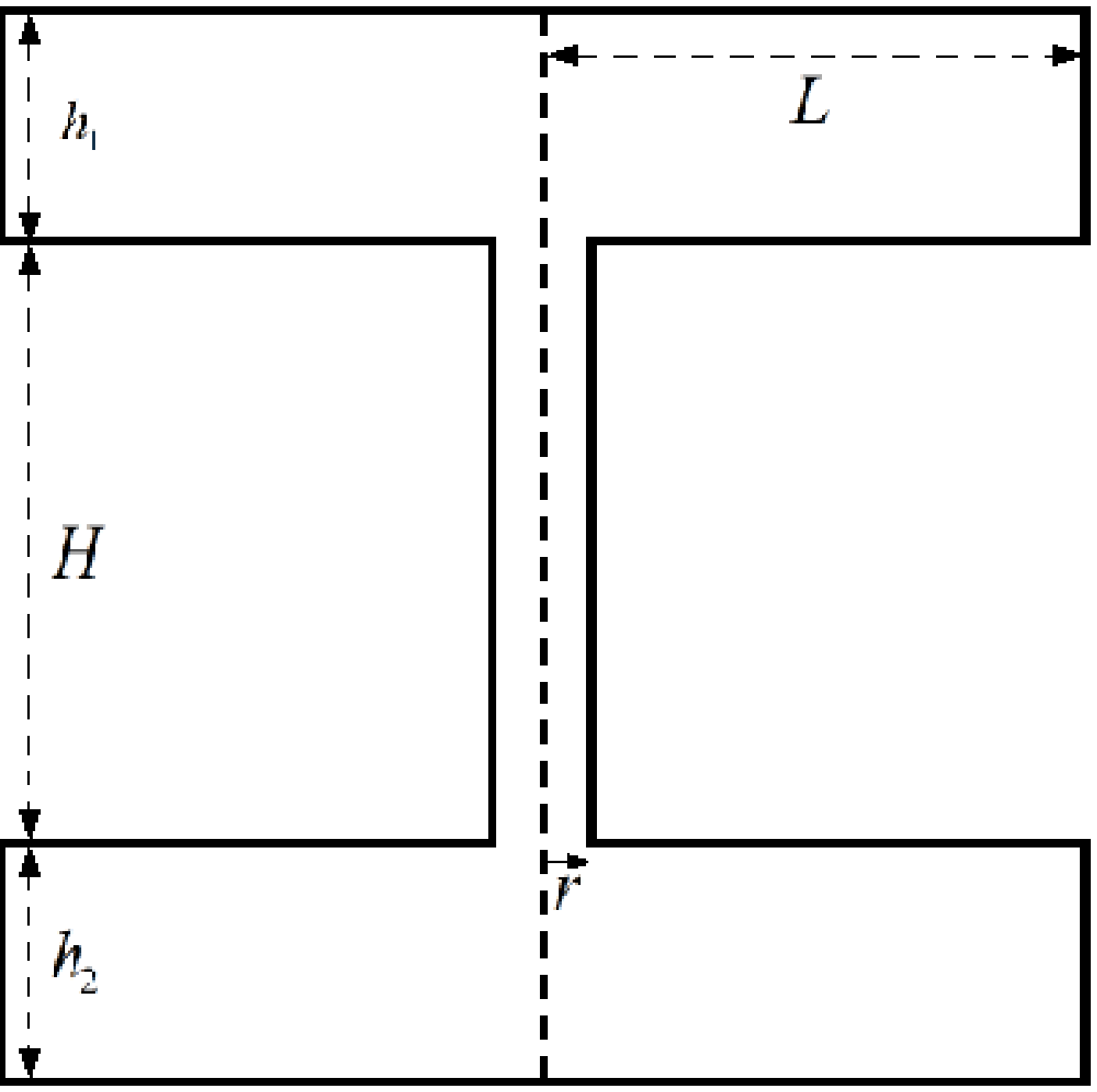}
  \includegraphics[width=2.7in]{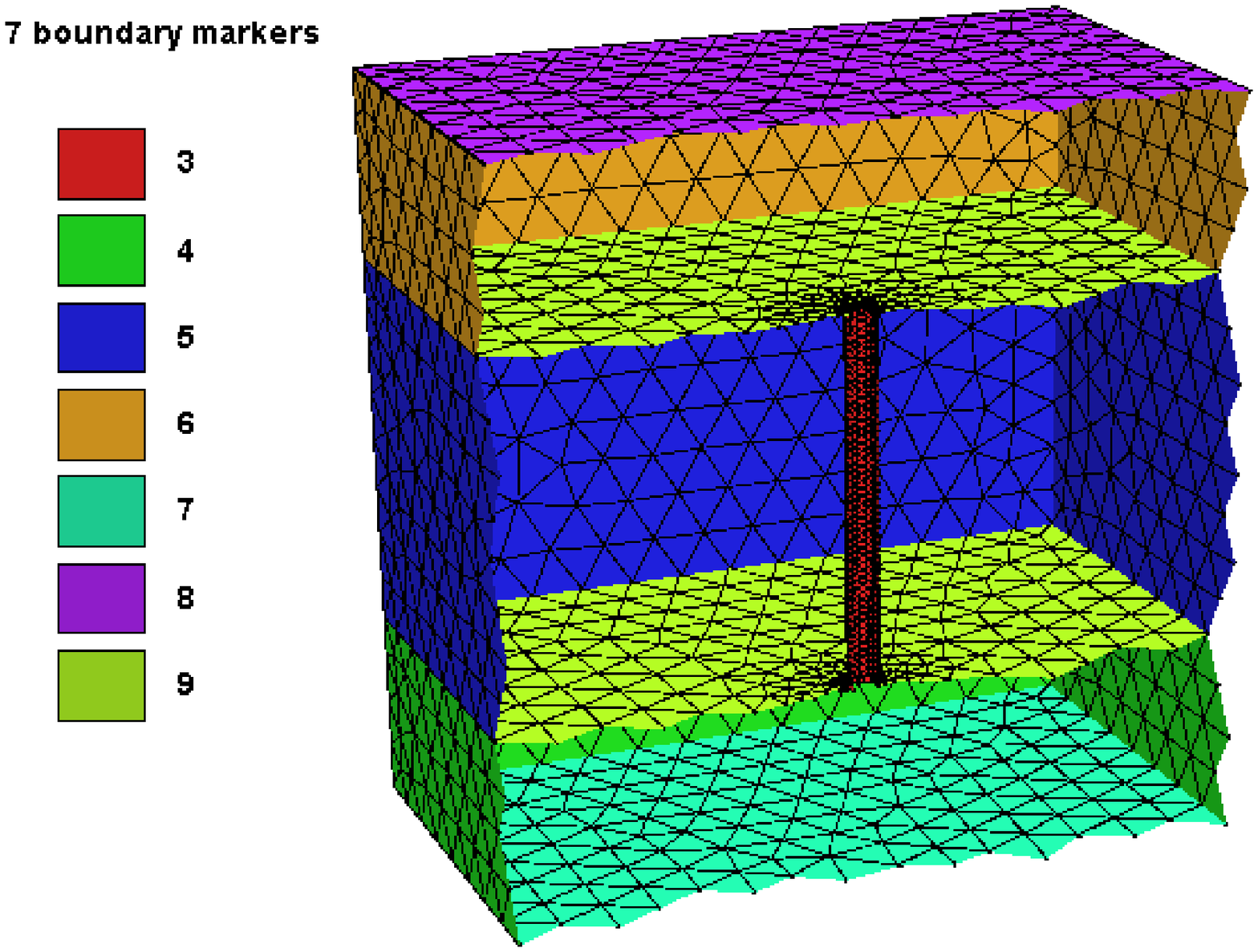}
  \caption{The geometry and mesh of the cylinder nanopore.}
%\centerline{\psfig{file=fig1_1.eps,width=1.5in}\psfig{file=fig2_1.eps,width=1.5in}}
%\centerline{\psfig{file=fig2_1.eps,width=1.5in}}
%\vspace*{8pt}
%\caption{A 2-D schematic view of the ionic solution system: (a) with one fixed biomolecule;
%(b) with an ion channel (or similar a nanopore) embedded in a membrane.}
\end{figure}
%\begin{figure}[htbp]
%  \centering
%  % Requires \usepackage{graphicx}
%  %\includegraphics[width=0.45\textwidth]{fig_vcmm_genemesh.jpg}
%  \includegraphics[width=1.9in]{fig3_1}
%  \includegraphics[width=2.7in]{fig3_2}
%  \caption{The geometry and mesh of the cylinder nanopore.}
%\end{figure}

The electrical current of the traditional PNP model across the pore can be calculated as:
\begin{equation*}
    I_{z} = -\sum_{i}q_{i}\int_{S}D_{i}(\frac{\partial c_{i}}{\partial z} + \frac{q_{i}}{k_{B}T}c_{i}\frac{\partial \phi}{\partial z})dxdy,
\end{equation*}
where S is a cut plane at any cross section inside the pore.

For the PNP model Eqs.~\eqref{2.34}-\eqref{2.35} from incomplete energy form \eqref{2.3}, the electrical current across the pore is calculated as:
\begin{equation*}
    I_{z} = -\sum_{i}q_{i}\int_{S}D_{i}(\frac{\partial c_{i}}{\partial z} + \frac{q_{i}}{k_{B}T}c_{i}\frac{\partial (\phi-\frac{1}{2}\phi_{D})}{\partial z})dxdy.
\end{equation*}
In the PNP model, the current can be split into two parts: the concentration diffusion part
\begin{equation*}
    I_{diff} = -\sum_{i}q_{i}\int_{S}D_{i}\frac{\partial c_{i}}{\partial z}dxdy,
\end{equation*}
and the potential drift part
\begin{equation*}
    I_{drift} = -\sum_{i}q_{i}\int_{S}D_{i}(\frac{q_{i}}{k_{B}T}c_{i}\frac{\partial \phi}{\partial z})dxdy.
\end{equation*}
The "screwed" PNP from incomplete energy form has a similar concentration diffusion part but a different potential drift part
\begin{equation*}
    I_{drift} = -\sum_{i}q_{i}\int_{S}D_{i}(\frac{q_{i}}{k_{B}T}c_{i}\frac{\partial (\phi-\frac{1}{2}\phi_{D})}{\partial z})dxdy.
\end{equation*}
\begin{figure}[pb]
\centerline{\psfig{file=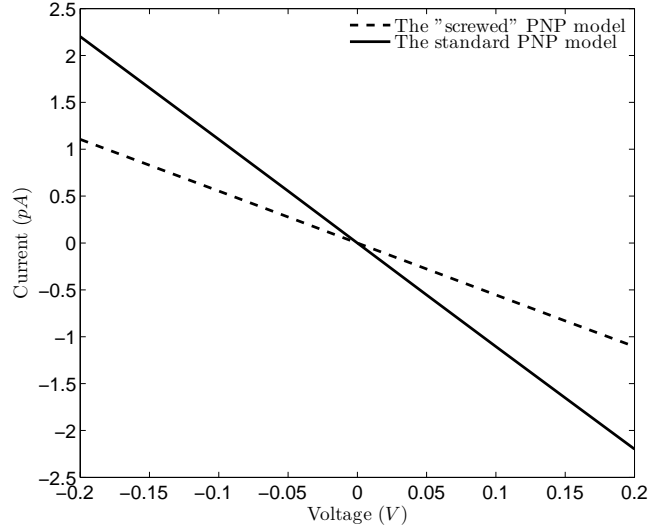,width=3.8in}}
\vspace*{8pt}
\caption{The current voltages characteristics obtained with the traditional (solid line) and "screwed" (dashed line) PNP models at bulk concentration 0.1M and membrane potential -0.2V.}
\end{figure}
Through comparison between the currents calculated by the PNP model and the "screwed" PNP model, it is observed that with such system setup the magnitude of current in the "screwed" PNP model derived from incomplete energy tends to be smaller than that in the traditional PNP model (see Fig. 4). The current resulted from the potential drift part is dominant compared to that from the concentration diffusion part (compare the order of magnitude in Figs. 5(a) and 5(b)). It is also observed that in the "screwed" PNP model , the potential drift part significantly underestimates the magnitude of the current, whereas the diffusion part exposes the opposite property.
\begin{figure}[bp]
\centering
%\subfigure[]{
%%\label{fig.1}
\includegraphics[width=3.6in]{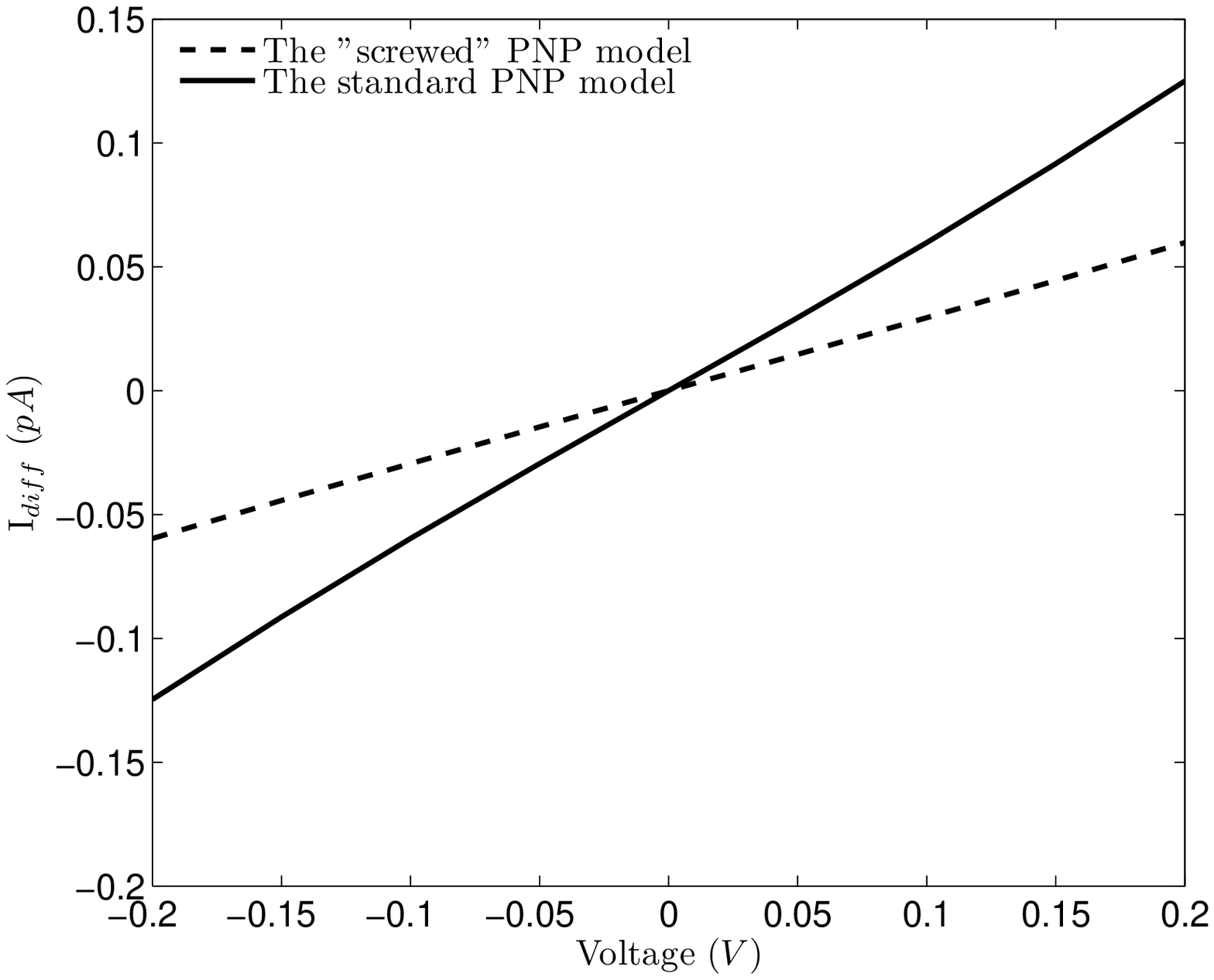}
%\subfigure[]{
%%\label{fig.2}
\includegraphics[width=3.6in]{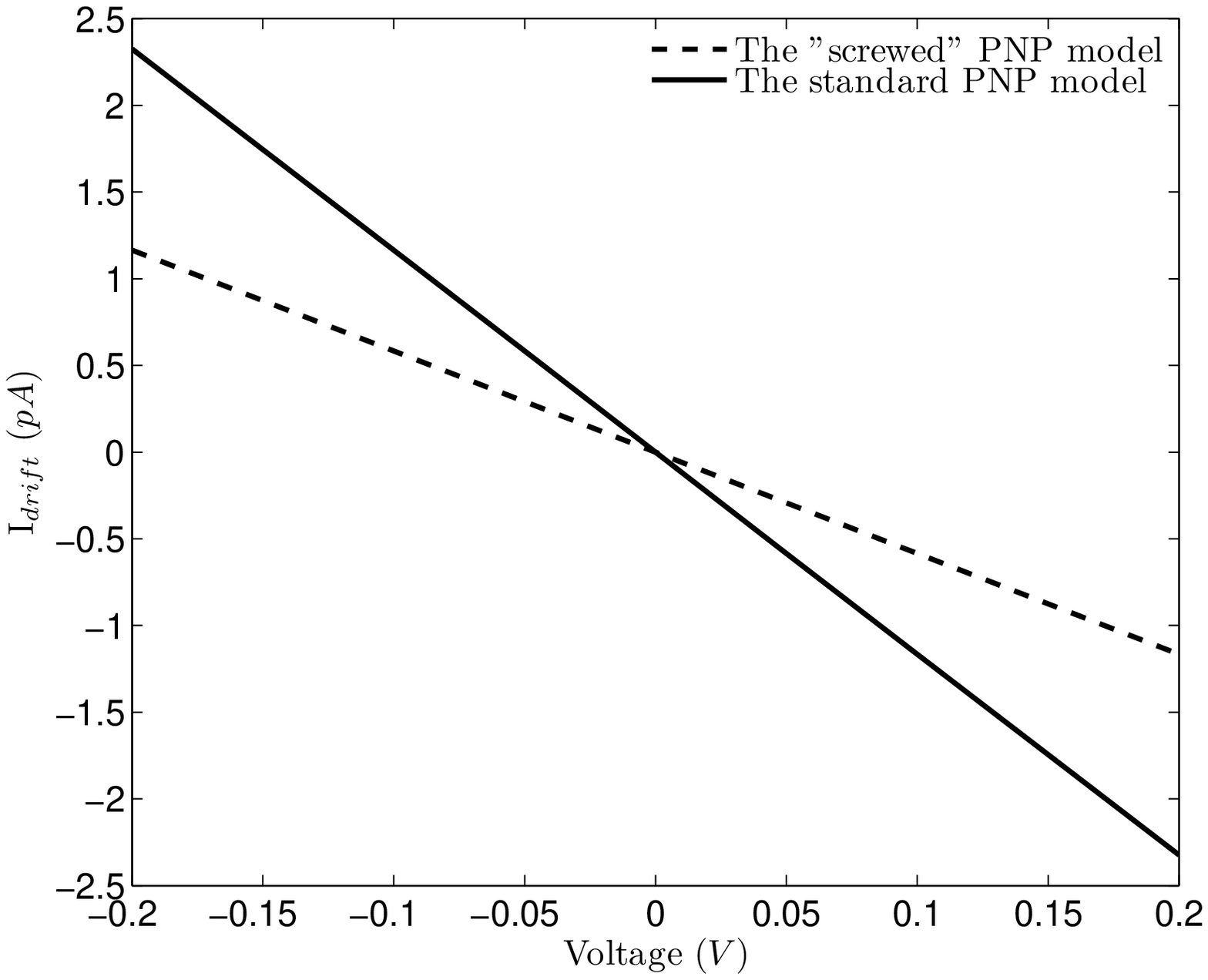}
\caption{Contribution of (a) the diffusion and (b) the drift parts of current in the traditional (solid line) and "screwed" (dashed line) PNP models.}
%%\label{Fig.one}
\end{figure}
%\newpage
\section{Conclusion}
In this paper, we present a mean field free energy functional of dielectrically inhomogeneous electrolyte solution in a finite domain with genetic Neumann/Dirichlet boundary conditions for potential.
In this new energy functional, the boundary interaction terms are physically reasonable, and are also crucial in mathematical analysis in order to consistently derive the correct PB and PNP equations.
We also show that in presence of non-zero Dirichlet boundary conditions for electric potential, the traditional energy form is not consistent with the traditional PB and PNP equations. Using variational method to the previous energy functional (usually by introducing corresponding homogeneous problem) may result in screwed (non-physical) Boltzmann distribution and PB/PNP models. Our numerical examples demonstrate the significant deviations of the results originated from the screwed models.
Furthermore, in a particular interesting case where the dielectric coefficient of the electrolyte solution depends on the local ionic concentrations, we derive the generalized PB and PNP equations from our complete free energy functional.
As for more complicate boundary conditions, it may be still an issue for free energy functional analysis.

\section*{Acknowledgments}
The authors thank Hanlin Li for helpful discussion.
This work was Supported by Science Challenge Project, No. JCKY2016212A503,
and China NSF (NSFC 91530102, NSFC 21573274).

\end{document}